\documentclass[10pt,usenatbib]{biom}
\usepackage{setspace}
\usepackage{lipsum}
\usepackage[normalem]{ulem}
\usepackage{natbib}
\usepackage{tikz}
\usepackage{amsmath}
\usepackage{float}
\usepackage{graphicx}
\usepackage{amsfonts}
\usepackage{amssymb}
\usepackage{colortbl}
\usepackage{booktabs}
\usepackage{multirow}
\usepackage[section]{algorithm}
\usepackage{multicol}
\usepackage{url,xcolor}
\usepackage{algorithmicx,setspace}
\usepackage{algpseudocode}
\usepackage{titlesec}

\titlespacing*{\section}{0pt}{0.5\baselineskip}{0.3\baselineskip}

\usepackage[font={footnotesize}]{caption}

\algblock{Input}{EndInput}
\algnotext{EndInput}
\algblock{Output}{EndOutput}
\algnotext{EndOutput}
\algblock{Update}{EndUpdate}
\algnotext{EndUpdate}

\def\bSig\mathbf{\Sigma}
\newcommand{\vect}[1]{\boldsymbol{#1}}

\title[Asymmetric predictability in causal discovery]{Asymmetric predictability in causal discovery: an information theoretic approach}

 \author{Soumik Purkayastha$^{*}$\email{soumikp@umich.edu},
 Peter X.-K. Song$^{**}$\email{pxsong@umich.edu} \\
 Department of Biotatistics, University of Michigan, Ann Arbor, MI, U.S.A.}

\begin{document}




\pagerange{\pageref{firstpage}--\pageref{lastpage}} 
\volume{ }
\pubyear{}
\artmonth{}




\label{firstpage}


\begin{abstract}
Causal investigations in observational studies pose a great challenge in research where randomized trials or intervention-based studies are not feasible. We develop an information geometric causal discovery and inference framework of "predictive asymmetry". For $(X, Y)$, predictive asymmetry enables assessment of whether $X$ is more likely to cause $Y$ or vice-versa. The asymmetry between cause and effect becomes particularly simple if $X$ and $Y$ are deterministically related. We propose a new metric called the Directed Mutual Information ($DMI$) and establish its key statistical properties. $DMI$ is not only able to detect complex non-linear association patterns in bivariate data, but also is able to detect and infer causal relations.  Our proposed methodology relies on scalable non-parametric density estimation using Fourier transform. The resulting estimation method is manyfold faster than the classical bandwidth-based density estimation. We investigate key asymptotic properties of the $DMI$ methodology and a data-splitting technique is utilized to facilitate causal inference using the $DMI$. Through simulation studies and an application, we illustrate the performance of $DMI$.

\end{abstract}

%



\maketitle


%

\section{Introduction}
\label{sec:intro} Discovery of causal relationships from observational data is a cornerstone of scientific research. Given bivariate observations $(X, Y)$, a fundamental question is whether $X$ causes $Y$ or, alternatively, $Y$ causes $X$. Even under simplifying assumptions of no confounding, no feedback loops, and no selection bias, a direct assessment of a bivariate causal relationship is a naturally hard problem. \citep{spirtes_2016}. 


In this paper we propose a new causal discovery statistic within the framework of Shannon's information theory \citep{Shannon1948} 
by means of a new statistic, the directed mutual information $(DMI)$.  $DMI$ serves as a statistical test for independence and also quantifies a notion of  ``asymmetric predictability'' or ``predictive asymmetry'' for bivariate $(X, Y)$. Not only are we able to estimate this new statistic in a computationally fast and robust manner, but also we establish a framework for statistical inference using a new data-splitting technique. 

Arguably, a key question in bivariate causal discovery methods is 
whether $X$ is the response variable and $Y$ is the predictor variable (or converse). Utilizing information theoretic measures, we propose the directed mutual information coefficients $DMI(X|Y)$ and $DMI(Y|X)$, which enable us to not only test for independence but also quantify ``predictive asymmetry'' for $(X, Y)$ - thereby establishing a sense of asymmetry in bivariate associations. 

Asymmetric predictability invokes the understanding that conditional entropy $H(Y \lvert X)$ measures the amount of uncertainty remaining in $Y$ after learning $X$. It is known that $H(X|Y) > \ (\text{or }<) \ H(Y|X)$ implies conditioning on $X$ and predicting $Y$ yields less (or more) uncertainty, thereby establishing a sense of ``asymmetric predictability''.


Our methodology development is motivated by a cohort study - the Early Life Exposures in Mexico to Environmental Toxicants (ELEMENT) cohort \citep{Hernandez_Avila_1996}. The study aims to investigate the direction of influence between DNA methylation (DNAm) alterations and cardiovascular outcomes. Specifically, in genes (namely, \emph{FGF5}, \emph{ATP2B1} and \emph{PRDM8}) that are linked with blood pressure (BP), researchers wish to investigate whether DNAm (specifically, cytosine-phosphate-guanine (CpG) methylation) status influences change in BP or if the converse is true \citep{dicorpo_2018}. Focusing on the CpG sites of the three candidate genes, we apply our $DMI$ framework to analyze data on DNAm and BP in the epigenetic ELEMENT study. Our method unveils predictive asymmetry exists between CpG sites and BP variation. This new finding provides a sense of directionality in an association study between BP variation and epigenetic biomarkers, paving the way for future advancements in individualized risk assessments and even therapeutic targets.


The organization of this paper is as follows. Section 2 presents the formulation of our framework. Section 3 presents the estimation methodology and theoretical guarantees. In Section 4, we present simulation studies used to evaluate the finite-sample performance of our method. We apply our method to analyze the motivating ELEMENT data example in Section 5. Finally, we make some concluding remarks in Section 6. Detailed proofs of the large sample properties are included in the Appendix.



\section{Directed Mutual Information $(DMI)$}
\label{sec:dmi}

\textbf{\emph{Mutual information is copula entropy.}} Let $X$ and $Y$ be two random variables with joint density function $f_{XY}$. Let $f_X$ and $f_Y$ be the marginal densities of $X$ and $Y$, respectively. The mutual information $MI(X, Y)$ \citep{Shannon1948} is:
\begin{equation} \label{eq:mi}
MI(X, Y) = E_{XY}\left\{ \log \frac{f_{XY}(X, Y)}{f_X(X) f_Y(Y)} \right\},
\end{equation} where $E_{XY}$ denotes expectation over $f_{XY}$. Some properties that make $MI$ an attractive measure of complex dependence include: (i) $MI \geq 0$ with equality if and only if $X$ and $Y$ are independent and (ii) a larger value of $MI$ indicates a stronger dependence between two variables.
We consider an equivalent formulation of $MI$ by invoking the marginal transformations $U = F_X(X) \sim \mathcal{U}(0, 1)$ and $V = F_Y(Y) \sim \mathcal{U}(0, 1)$, where $F_X$ and $F_Y$ are the cumulative distribution functions (CDF) of $X$ and $Y$ respectively. According to the Sklar's theorem \citep{Sklar1959FonctionsDR}, we have $MI(X, Y) = MI(U, V) = E_{\vect{Z}}\left\{ \log c(\vect{Z}) \right\}$, where $\vect{Z} = (U, V) \in \left[0, 1 \right]^2$ and $c(\vect{z})$ is the unique copula density function defined on the unit square $[0, 1]^2$. Since we focus only on the joint copula density $c$, instead of $f_{XY}, f_X$ and  $f_Y$, the computational burden of estimation methods based on previous understanding of $MI$ is greatly reduced \citep{Ma2011}. Assuming knowledge of the estimator $\hat{c}$ of $c$, $MI$ is the sample mean of the log of $\hat{c}$ at (transformed) data points $\left[ \vect{Z}_j =  \left\{\hat{F}_{X}(X_j), \hat{F}_{Y}(X_j) \right\}\right]_{j=1}^n$ on the compact domain $[0, 1]^2$, where $\hat{F}_X$ and $\hat{F}_Y$ are the empirical CDFs (ECDF) of $X$ and $Y$ respectively.

\textbf{\emph{Marginal and conditional entropies.}} The marginal entropy of $X$ is defined as $H(X) = E_X\left\{ - \log f_X(X) \right\},$ where $E_{X}$ denotes expectation over $f_X$, while the conditional entropy of $X$ conditioned on $Y$ is given by $H(X|Y) = E_{XY} \left\{ - \log f_{XY}(X, Y)/f_Y(Y) \right\}.$ These quantities are related to $MI(X, Y)$ via the following identity:
\begin{equation}
\label{eqn:identity} 
\begin{aligned}
H(X,Y) &= MI(X, Y) + H(X|Y) + H(Y|X). 
\end{aligned}
\end{equation} where $H(X, Y) = E_{XY} \left\{- \log f_{XY}(X, Y) \right\}$ is the joint entropy of $(X, Y)$. 
If we assume knowledge of the estimator $\hat{f}_X$, estimation of $H(X)$ is equivalent to evaluating the sample mean of the log of the estimator $\hat{f}_X$ at data points $\left\{ {X}_j \right\}_{j=1}^n$. 

\textbf{\emph{Entropy ratio.}} The entropy decomposition in Equation \ref{eqn:identity} resembles Fisher's seminal decomposition of the total variation into the sum of both within and between variations in the analysis of variance (ANOVA). Equation \ref{eqn:identity} motivates the development of a metric to quantify asymmetric predictability: the total entropy $H(X, Y)$ may be decomposed to establish a sense of ``symmetric  behaviour'' through $MI(X, Y)$ and ``asymmetric behaviour'' through $H(X|Y)$ and $H(Y|X)$. Under our proposed tenet of asymmetric predictability, an asymmetry between $X$ and $Y$ emerges if $H(Y|X) \neq H(X|Y)$. Comparing $H(X|Y) > \ (\text{or }<) \ H(Y|X)$ reveals if conditioning on $X$ and predicting $Y$ yields less (or more) uncertainty.  As a result, it highlights which of $X$ or $Y$ has a more dominant predictive role to play in a bivariate relationship.
We define the entropy ratio of $X$ relative to $Y$, $ER(X|Y)$:
\begin{equation}
    \label{eq:relcondent}
    ER(X|Y) = \frac{\exp\{H(X|Y)\}}{ \left[\exp\{H(X|Y)\} + \exp\{H(Y|X)\} \right]},
\end{equation}
where the exponential transformation guarantees all components of $ER(X|Y)$ are positive. Note that (i) $ER(X|Y) = ER(Y|X) = 1/2$ if and only if $H(X|Y) = H(Y|X)$ and (ii) $ER(X|Y) \lessgtr 1/2$ implies $H(X|Y) \lessgtr H(Y|X)$, with $H(Y|X) < H(X|Y)$ establishing $X$ as the ``dominant predictor variable'' that exerts more ``influence'' on $Y$. 

The $ER$ has a nice transitive property: if we assume $X$ is more predictive than $Y$, i.e., $ER(X|Y) > 1/2$ and $Y$ is more predictive than $Z$, i.e., $ER(Y|Z) > 1/2$, a little algebra yields $ER(X|Z) > 1/2$, i.e., an ordering of $ER$ between $(X, Y)$ and $(Y, Z)$ provides insight on asymmetric predictability in a third pair, i.e., $(X, Z)$ within our proposed framework. 


\textbf{\emph{Directed mutual information.}} We now propose a new measure called the \textit{directed mutual information} ($DMI$) that can detect dependence between two random variables and also test for asymmetric predictability  between $X$ and $Y$. The measure is defined by:
\begin{equation}
    \label{eq:dmi}
    DMI\left(X\middle| Y\right) = MI(X, Y) \times ER(X|Y). 
\end{equation} Note that $DMI$ is not a symmetric measure, i.e., $DMI(X|Y)$ and $DMI(Y|X)$ are not necessarily identical. Intuitively, our new estimator $DMI(X|Y)$ is a scaled function of $MI$, which is an effective tool to capture symmetric association. The scaling factor $ER(X|Y)$ aims to capture asymmetric behaviour between $X$ and $Y$ by comparing their relative predictive performance when they are associated. Some properties of $DMI$ include (i) $DMI\left(X\middle| Y\right) \geq 0$ and $DMI\left(Y\middle| X\right) \geq 0$, with equality if and only if $X$ and $Y$ are independent and (ii) $DMI\left(X\middle| Y\right) \lessgtr DMI\left(Y\middle| X\right)$ implies $ER\left(X\middle| Y\right) \lessgtr ER\left(Y\middle| X\right)$. While (i) forms the basis for testing statistical independence, under our proposed tenet of asymmetric predictability, (ii) naturally establishes whether $X$ or $Y$ is the ``dominant predictor variable''. We define 
\begin{equation}
    \label{eq:del}
    \Delta = DMI\left(X\middle| Y\right) - DMI\left(Y\middle| X\right),
\end{equation} where $\Delta > 0$ (or $\Delta < 0$) establishes $X$ as the dominant predictor variable over $Y$, resulting in a certain asymmetry between the two variables (or $Y$ as the dominant predictor variable over $X$). 

\section{Estimation and inference for $DMI$ and $\Delta$}
\label{sec:methodology}
In order to estimate the $DMI$, we first tackle the problem of estimating a continuous density from a set of bivariate data points. Kernel density estimators (KDEs) are commonly used techniques for estimating the probability density function (PDF). The KDE method requires specification of some kind of bandwidth or a kernel function with a specific bandwidth. Bandwidth tuning is a tricky issue  \citep{Silverman1986} and is computationally intensive, since it requires repeated density estimation.  The difficulty compounds for higher dimensions \citep{Duong2005}. 
A review of automatic selection methods \citep{Heidenreich2013} recommends a variety of different methods, depending on data set characteristics (including sample size, distribution smoothness, and skewness) and thus, is hard to implement in practise.

An attractive alternative to the kernel-based estimation strategy involves Fourier transformation; \cite{Bernacchia2011} present a estimator for the univariate case while \cite{OBrien2016} extend the method to higher dimensions. This estimation process is demonstrably orders of magnitude faster than comparable, state-of-the-science density estimation packages in R \citep{OBrien2016} while maintaining comparable statistical error performance. This extension provides a data-driven bivariate PDF estimator that is optimal, fast, and unencumbered by the need for user-selected parameters. The estimator is called the self-consistent estimator (SCE), and is used to estimate $DMI$ and $\Delta$.  

\textbf{\emph{The self-consistent density estimator.}} Let us consider a random bivariate sample denoted by $\mathcal{S} = \{\vect{X}_j: \vect{X}_j \in \mathcal{X}, j = 1,2,\ldots,n\}$ from density $f$ with support $\mathcal{X}$ (without loss of generality, $\mathcal{X} = \mathbb{R}^2$). We assume $f$ belongs to the Hilbert space of square integrable functions, given by $\mathcal{L}^2 = \left\{f: \int f^2(\vect{x})d\vect{x} < \infty \right\}.$
We consider the bivariate self-consistent estimator (SCE)  $\hat{f} \in \mathcal{L}^2$. In order to define $\hat{f}$, we require a kernel function $K$, which belongs to the class of functions given by
\begin{equation*}
    \mathcal{K} := \left\{K: K(\vect{x}) \geq 0, K(\vect{x}) = K(-\vect{x}) \ \forall\ \vect{x}; \int K (\vect{t})d\vect{t} = 1 \right\}. 
\end{equation*} Specifically, $\hat{f}$  is given by the convolution of a kernel $K$ and the set of delta functions centered on the dataset as follows: 
\begin{equation}
\begin{aligned}
 \label{eqn:kde}
     \hat{f}(\vect{x}) &\equiv {n}^{-1} \sum_{j = 1}^{n} K(\vect{x} - \vect{X}_j) \\
     &=  {n}^{-1} \sum_{j = 1}^{n} \int_{\mathcal{X}} K(\vect{s}) \delta(\vect{x} - \vect{X}_j - \vect{s})d\vect{s}, \quad \vect{x} \in \mathcal{X} 
\end{aligned} 
\end{equation} where $\delta(\vect{x})$ is the Dirac delta function \citep{Kreyszig_Erwin2020-07-21}. 
Our aim is identify the optimal kernel $\hat{K}$, where ``optimality'' is intended as minimising the mean integrated square error (\emph{MISE}) between the true density $f$ and the estimator $\hat{f}$:
\begin{equation}
\begin{aligned}
 \label{eqn:mise}
     \hat{K} &= \underset{{K} \in \mathcal{K}}{\text{argmin}} \ \text{\emph{MISE}}(\hat{f}, f) \\ &= \underset{{K} \in \mathcal{K}}{\text{argmin}} \  \mathbb{E}\left[\int_{\mathbb{R}^2}\{\hat{f}(\vect{x})-f(\vect{x})\}^{2} d\vect{x}\right],
\end{aligned}
\end{equation} where the $\mathbb{E}$ operator denotes taking expectation over the entire support of $f$. The SCE in Equation \ref{eqn:kde} may be represented equivalently by its inverse Fourier transform pair, $\hat{\phi} \in \mathcal{L}^2$:
\begin{equation}
\begin{aligned}
 \label{eqn:kde_fouriertransformed}
     \hat{\phi}(\mathbf{t}) &= \mathcal{F}^{-1} \big(\hat{f}(\mathbf{x})\big)  = {\kappa}(\mathbf{t}) \mathcal{C}(\mathbf{t}),
\end{aligned}
\end{equation} where $\mathcal{F}^{-1}$ represents the multidimensional inverse Fourier transformation from space of data $\mathbf{x} \in \mathbb{R}^2$ to frequency space coordinates $\mathbf{t} \in \mathbb{R}^2$. ${\kappa} = \mathcal{F}^{-1} \big({K}\big)$ is the inverse Fourier transform of the kernel ${K}$ and $\mathcal{C}$ is the empirical characteristic function (ECF) of the data, defined as
\begin{equation}
\begin{aligned}
 \label{eqn:ecf}
     \mathcal{C}(\mathbf{t}) &= {n}^{-1} \sum_{j = 1}^{n} \exp{(i \mathbf{t}^\prime \mathbf{X}_j)}.
\end{aligned}
\end{equation}
\cite{Bernacchia2011} derive the optimal transform kernel $\hat{\kappa}$ that minimizes the $MISE$ given by Equation \ref{eqn:mise}, given as follows:
\begin{equation}
\begin{aligned}
 \label{eqn:kernel_fouriertransformed}
 \hat{\kappa}(\mathbf{t}) &= \frac{{n}}{2(n-1)} \left[1 + \sqrt{1 - \frac{4({n}-1)}{ \lvert n \mathcal{C}(\mathbf{t}) \rvert^2 }} \right] I_{A_{n}}(\mathbf{t}),
\end{aligned}
\end{equation} where $A_{n}$ serves as a low-pass filter that yields a stable estimator (see Remarks \ref{remark_filter1} and \ref{remark_filter2} in the appendix). We follow the nomenclature of \cite{Bernacchia2011} and denote $A_{n}$ as the set of ``acceptable frequencies''. The optimal transform kernel $\hat{\kappa}$ in Equation \ref{eqn:kernel_fouriertransformed} may be antitransformed back to the real space to obtain the optimal kernel $\hat{K} \in \mathcal{K}$, which yields the optimal density estimator $\hat{f}$ according to Equation \ref{eqn:kde}.
Theorem \ref{thm:strong_conv} presents the sufficient conditions for the estimate $\hat{f}$ to converge to the true density $f$ for ${n} \rightarrow \infty$. 

\begin{theorem}
\label{thm:strong_conv}
Let the true density $f$ be square integrable and its corresponding Fourier transform $\phi$ be integrable, then the self consistent estimator $\hat{f}$, which is defined by Equations \ref{eqn:kde} - \ref{eqn:kernel_fouriertransformed} converges almost surely to the true density as $n \rightarrow \infty$, under the additional  assumptions $\mathcal{V}(A_{n}) \rightarrow \infty, \mathcal{V}(A_{n})/\sqrt{{n}} \rightarrow 0$ and $\mathcal{V}(\bar{A}_{n}) \rightarrow 0$ as $n \rightarrow \infty$. Further, assuming $f$ to be continuous on dense support $\chi$, we have uniform almost sure convergence of $\hat{f}$ to $f$ as $n \rightarrow \infty$. 
\end{theorem} Here $\bar{A}_{n}$ is the complement of $A_{n}$ and the volume of $A_n$ is given by $\mathcal{V}(A_{n})$.  
Using Theorem \ref{thm:strong_conv}, Theorem \ref{thm:strong_conv_mi} establishes almost sure convergence of $\widehat{MI}$ to $MI$ as $n \rightarrow \infty$. 
\begin{theorem}
\label{thm:strong_conv_mi}
Let the conditions presented in Theorem \ref{thm:strong_conv} hold. We assume that the true copula density $c$ and marginal densities $f_X$ and $f_Y$ are smooth and bounded away from zero and infinity on their respective support. Under these assumptions, we have 
\begin{equation*}
    \widehat{DMI}(X|Y) - DMI(X|Y)  \overset{\text{a.s.}}{\rightarrow} 0, \text{ as }n \rightarrow \infty.
\end{equation*} 
\end{theorem}

\textbf{\emph{Asymptotic behavior of $\widehat{DMI}$ and $\hat{\Delta}$.}} We establish the asymptotic normality of the resulting estimators $\widehat{DMI}$ and $\hat{\Delta}$ by proposing a new data-splitting inference as described below. Note that the invocation of data-splitting is solely for statistical inference and not required for estimation purposes alone.

We split the data of size $n$ into two disjoint sets, $\mathcal{D}_1$ and $\mathcal{D}_2$ with sample sizes $n_1$ and $n_2$ respectively ($n_1 + n_2 = n$). As described in Section \ref{sec:methodology}, using data from $\mathcal{D}_1$ we obtain SCEs of the copula density function $\hat{c}_{\mathcal{D}_1}$ and the marginal density functions $\hat{f}_{X; \mathcal{D}_1}$ and  $\hat{f}_{Y; \mathcal{D}_1}$. Under some mild assumptions (see Theorem \ref{thm:strong_conv}), we have proved uniform almost sure convergence of these estimators to their population counterparts on their respective support sets as sample size $n_1 \rightarrow \infty$. 
Using data from $\mathcal{D}_2$, we then evaluate the estimators 
\begin{gather}
  \begin{aligned}
\label{eq:dmi_parts_estim}
    \widehat{MI}(X, Y) &= {n_2}^{-1}\sum_{j=1}^{n_2}  \log \left\{\hat{c}_{\mathcal{D}_1}(\vect{Z}^{\mathcal{D}_2}_j)\right\},\\
    \hat{H}(X) &=  {n_2}^{-1}\sum_{j=1}^{n_2} - \log \left\{\hat{f}_{X; \mathcal{D}_1}(X^{\mathcal{D}_2}_j)\right\},\\
    \hat{H}(Y) &= {n_2}^{-1}\sum_{j=1}^{n_2}  - \log \left\{\hat{f}_{Y; \mathcal{D}_1}(Y^{\mathcal{D}_2}_j)\right\},
\end{aligned}  
\end{gather}
Consequently, we obtain the following estimators $\widehat{DMI}(X|Y) = \widehat{MI}(X, Y)  \times  \widehat{ER}(X|Y)$, $\widehat{DMI}(Y|X) = \widehat{MI}(X, Y) \times \widehat{ER}(Y|X)$ and $\hat{\Delta} = \widehat{DMI}(X|Y) - \widehat{DMI}(Y|X)$.
The following two theorems establish asymptotic normality of $\widehat{DMI}$ and  $\hat{\Delta}$ as $n_1 \wedge n_2 := \min(n_1, n_2) \rightarrow \infty$.

\begin{theorem}
\label{thm:dmi_normality}
Let the conditions presented in Theorem \ref{thm:strong_conv_mi} hold. 
Further, we assume $MI \neq 0$. Under these assumptions, we have $ \sqrt{n_2}\left\{\widehat{DMI}(X|Y) - DMI(X|Y) \right\} \overset{\mathcal{D}}{\rightarrow} N \left(0, \sigma_{DMI}^2\right)$ as $n_1 \wedge n_2 \rightarrow \infty$, with $\sigma_{DMI}^2$ denoting the asymptotic variance of the estimator.  The appendix presents a closed-form expression of $\sigma_{DMI}^2$, which may be estimated using standard Monte Carlo tools \citep{robert_2010}.
\end{theorem}

\begin{theorem}
\label{thm:del_normality}
Let the conditions presented in Theorem \ref{thm:dmi_normality} hold. Under these assumptions, we have $\sqrt{n_2}\left(\hat{\Delta} - {\Delta} \right) \overset{\mathcal{D}}{\rightarrow} N(0, \sigma_{\Delta}^2)$ as $n_1 \wedge n_2 \rightarrow \infty$, with $\sigma_{\Delta}^2$ denoting the asymptotic variance of the estimator. The appendix presents a closed-form expression of $\sigma_{\Delta}^2$, which may be estimated using standard Monte Carlo tools \citep{robert_2010}.
\end{theorem}

\begin{remark}
Note that the sample size-based scaling factor associated with both $\hat{\Delta}$ and $\widehat{DMI}$ are linked with the sample size $n_2$ of the second data split $\mathcal{D}_2$, rather than the entire combined sample $n = n_1 + n_2$. This is a consequence of the data splitting method we propose. However, we require  $n_1 \wedge n_2 \rightarrow \infty$ while analysing the asymptotic behaviour of $\hat{\Delta}$ and $\widehat{DMI}$. One way to split the data would be to create two datasets of (approximately) similar size, i.e., $(n_1 = n_2) \approx n/2$. 
\end{remark}

\textbf{\emph{Testing for independence using $DMI$.}} The  statistic $\widehat{DMI}(X|Y) = \widehat{MI}(X, Y) \times \widehat{ER}(X|Y)$ is used to test for independence between two variables $X$ and $Y$. Rejection rules of the test based on the asymptotic distributions require data with large sample sizes, which may not be always available in practice. To ensure a stable and reliable performance, we implement a permutation-based test as it can give a precise finite-sample distribution of the test statistic for even small samples. 
With the null hypothesis of independence rejected, we may further test for the directionality of dependence as described in the next section.

\textbf{\emph{Test for asymmetric predictability using $\Delta$.}} The difference $\Delta = {DMI}(X|Y) - {DMI}(Y|X)$ determines a direction of dependence between $X$ and $Y$. 
When $X$ and $Y$ are not independent, with both $DMI(X|Y) > 0$ and $DMI(Y|X) > 0$, the null hypothesis $H_0: \Delta = 0$ signifies a bivariate relationship with ``predictive symmetry''. If $\Delta$ is significantly larger (or smaller) than zero, then we assign a direction of dependence from $Y$ to $X$ (or $X$ to $Y$) since $X$ (or $Y$) exerts ``predictive dominance'' on $Y$ (or $X$). We propose an asymptotic test based on the large sample behaviour of $\hat{\Delta} = \widehat{DMI}(X|Y) - \widehat{DMI}(Y|X)$ to test for the null hypothesis $H_0$ described above. Using the result presented in Theorem \ref{thm:del_normality}, we obtain a $95\%$ asymptotic confidence interval (CI) of $\hat{\Delta}$.  If the CI contains zero, we claim $X$ and $Y$ have predictive symmetry. If the CI lies to the right (left) of zero, we acquire data evidence in favor of  $DMI(X|Y)$ being significantly larger (or smaller) than $DMI(Y|X)$. Consequently, we assert at $5\%$ level of significance that $X$ exerts predictive dominance over $Y$ (or $Y$ exerts predictive dominance over $X$) in their bivariate asymmetric relationship. 

\section{Simulation studies}

\textbf{\emph{I: Testing for independence.}}
Let $S = \left\{(X_j, Y_j) \right\}_{j=1}^n$ be a random sample of $n$ observations drawn from a bivariate PDF $f_{XY}$ on $\mathbb{R}^2$. 
In each pattern described below, the signal parameter $a$ determines the strength of association between $X$ and $Y$ with $a = 0$ denoting independence. Upon increasing $a$, we increase the signal strength of $X$ in $Y$ relative to the independent noise $\epsilon$, implying a departure from independence (i.e. the null case):   
\begin{enumerate}
    \item[(P1)] \textbf{Linear:} 
    $X \sim N(\mu = 0, \sigma = 1)$ and $Y = aX + \epsilon$; $\epsilon \sim N(0, 0.5)$.
    \item[(P2)] \textbf{Quadratic:} 
    $X \sim N(0, 1)$ and $Y = aX^2 + \epsilon$; $\epsilon \sim N(0, 0.5)$.
    \item[(P3)] \textbf{Circular:} 
    $\theta \sim U(0, 1), X = 3\cos(2\pi\theta)$ and  $Y = 3a\sin(2\pi\theta) + \epsilon$; $\epsilon \sim N(0, 1)$.
    \item[(P4)] \textbf{Spiral:} 
    $\theta \sim U(0, 4), X = a\theta\cos(\pi\theta) + \epsilon_1$ and  $Y = \theta \sin(\pi\theta) + \epsilon_2$; $\epsilon_1, \epsilon_2 \sim N(0, 0.1)$.  
    \item[(P5)] \textbf{Exponential:} 
    $U \sim U(-3, 3), X = U + \epsilon_1$ and  $Y = a\exp(U) + \epsilon_2$; $\epsilon_1, \epsilon_2 \sim N(0, 0.1)$. 
    \item[(P6)] \textbf{Sinusoidal:}
    $U \sim U(0, \sqrt{12}), X = U + \epsilon_1$ and $Y = a\sin(2\pi U /\sqrt{12}) + \epsilon_2$; $\epsilon_1 , \epsilon_2 \sim N(0, 0.1)$. 
\end{enumerate} 
In each pattern above, the systematic components are independent of the error components. For each scenario described above, we examine the performance of the $DMI$-based permutation test for independence through both type I error rate under the null case, i.e., independence, with $a=0$ and the empirical power curve for different values of $a$ summarized from $r = 1000$ rounds of simulation. We vary sample sizes $n \in \left\{250, 500, 1000\right\}$. In Figure \ref{fig:simulation1_scatter}, we present scatter plots of each of six patterns for specific values of $a \neq 0$, exhibiting various association patterns from linear to highly nonlinear relationships.

The simulation results are presented in Figure \ref{fig:simulation1}. Note that the test returns a size that is approximately at the nominal $\alpha = 0.05$ under the null hypothesis of independence, exhibiting control over Type I error. As shown in Figure~\ref{fig:simulation1}, in each scenario described above, increasing the strength of association in $(X,Y)$ results in increased empirical power of the test. 
Moreover, the empirical power increases as we increase sample size in all six cases.

\begin{figure*}
\centerline{\includegraphics[width = 0.75\textwidth, height = 6in]{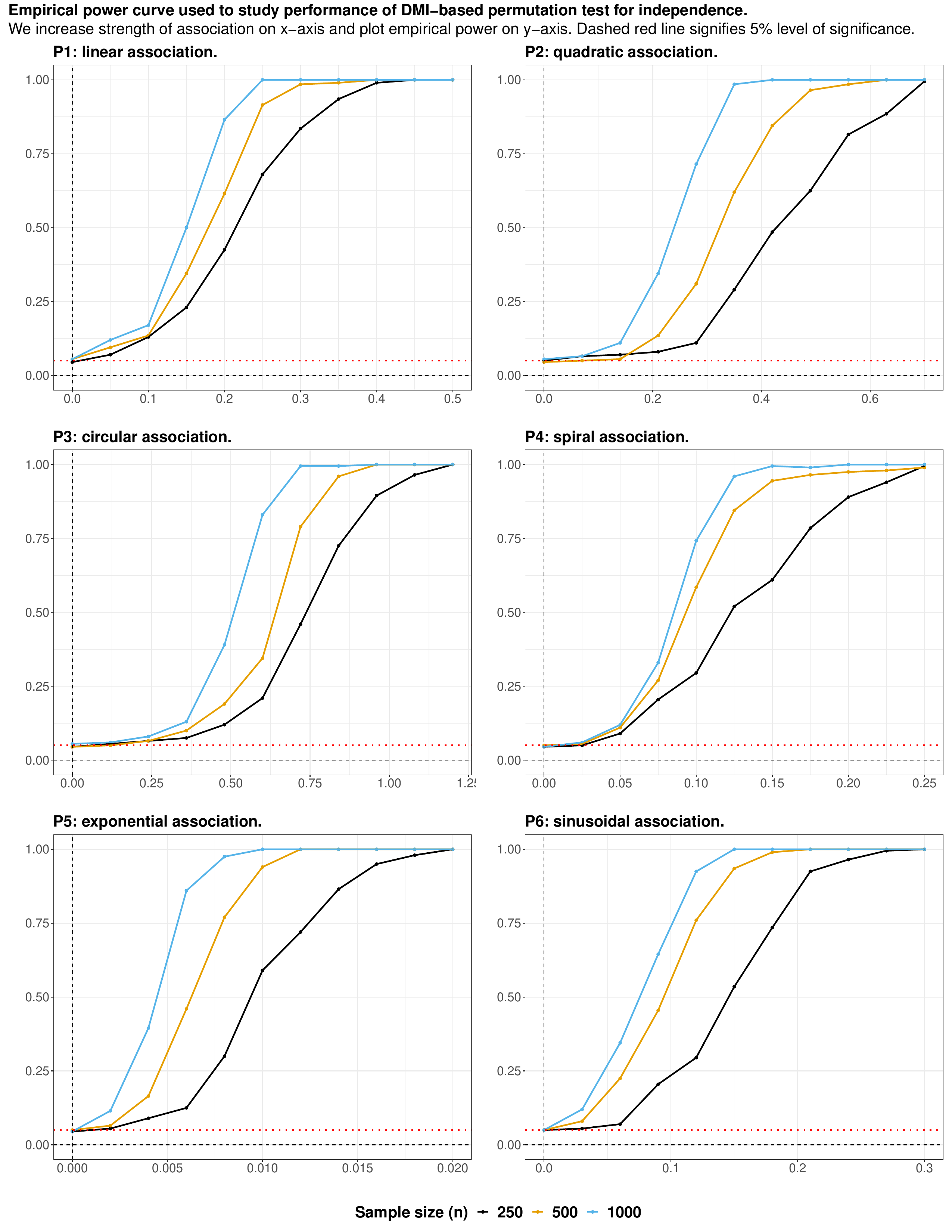}}
\caption{Power curves for simulation study I: testing for independence using $DMI$.}
\label{fig:simulation1}
\end{figure*}

\textbf{\emph{II: Testing for predictive asymmetry.}}
We analyze the test for predictive asymmetry, which is the main focus of this paper. Using the statistic $\hat{\Delta}(X|Y) =  \widehat{DMI}(X|Y) - \widehat{DMI}(Y|X)$, we want to investigate departure from symmetric associations in synthetic bivariate datasets. We generate a sample of $n$ observations drawn from a bivariate PDF $f_{XY}$ on $\mathbb{R}^2$, through representation involving Sklar's theorem  \citep{Sklar1959FonctionsDR} the underlying copula density function and the two associated marginal densities, given as follows:
\begin{enumerate}
    \item \textbf{Choice of copula density:} we choose the bivariate Gaussian copula with $\rho = 0.5$, the Clayton copula with $\theta = 2$ and the Gumbel copula with $\theta = 2$. For more details on various copula classes, please see \cite{czado_2019, joe_2014, nelsen_2006}. 
    \item \textbf{Choice of marginals:} choice of the marginal densities influences predictive asymmetry in the proposed information-theoretic framework. We consider the symmetric Gaussian density function $N(0, \sigma)$, the positively skewed $\text{Exp}(\text{rate} = \lambda)$, and the negatively skewed $\text{Log-Exponential}(\text{scale} = \gamma)$. We vary the parameters $\sigma$, $\lambda$ and $\gamma$ over a range of values. 
\end{enumerate}
  Note that by increasing the rate parameter $\lambda$ for the exponential distribution as well as the scale parameter $\gamma$ for the log exponential distribution yield decreased marginal (and conditional) entropy values. In contrast, increasing $\sigma$ results in increased entropy values in the Gaussian marginal choice. 
  For a given association pattern 
  we generate $R = 1000$ bivariate $i.i.d.$ samples on $(X, Y)$, each of size $n = 500$, in order to draw summary statistics. For each simulated sample, we study the estimate $\left\{\hat{\Delta}_r(X|Y)\right\}_{r=1}^R$ and compute the mean and quantile-based $95\%$ confidence interval of  $\hat{\Delta}_r$.  Based on Theorem \ref{thm:del_normality}, we note that the simulation-based estimates and theoretical values of $95\%$ confidence interval align closely with one another in all simulation cases considered. In Figure \ref{fig:simulation2}, we present a density heatmap for a specific combination of copula and marginal density parameters, in addition to estimated $\hat{\Delta}$ for different marginal settings.

We note that $\hat{\Delta}$ captures departure from association with ``predictive symmetry'' in bivariate datasets: our simulation results vary based on how we specify the marginal (and hence, conditional) density functions and not on the underlying copula family. Subplots (A1), (A2), (A3), and (A4), display density heatmaps of some of the bivariate distributions considered. Subplots (B1), (B2), (B3), and (B4) examine behaviour of $\hat{\Delta}$ upon changing the marginal parameters which influence asymmetric predictability in our framework. In (B1) and (B3), both marginals are either exponential or log-exponential. Note that increasing (or decreasing) the $X$-marginal parameter while keeping the $Y$-marginal fixed causes the conditional entropy of $X$ relative to $Y$ to decrease (or increase), since increasing $\lambda$ for exponential and $\gamma$ for log-exponential distributions causes the associated entropy to decrease. In (B2), for $\lambda = 0.75$ and $\sigma = 0.88$, we report a symmetric bivariate association. Upon increasing (or decreasing) $\sigma$ we note an increase (or decrease) in the marginal entropy of $X$ and hence a positive (or negative) $\hat{\Delta}$, implying predictive dominance of $X$ over $Y$ ($Y$ over $X$), since increasing $\sigma$ for normal distributions causes the associated entropy to increase. Similarly, in (B4), for $\gamma = 0.88$ and $\sigma = 1.25$, we report a balanced bivariate association. Upon increasing (or decreasing) $\sigma$ we note an increase (or decrease) in the conditional entropy of $X$ relative to $Y$ and hence a positive (or negative) $\hat{\Delta}$ implying predictive dominance of $X$ over $Y$ (or $Y$ over $X$).

\begin{figure*}
\centerline{\includegraphics[width = \textwidth, height = 4in]{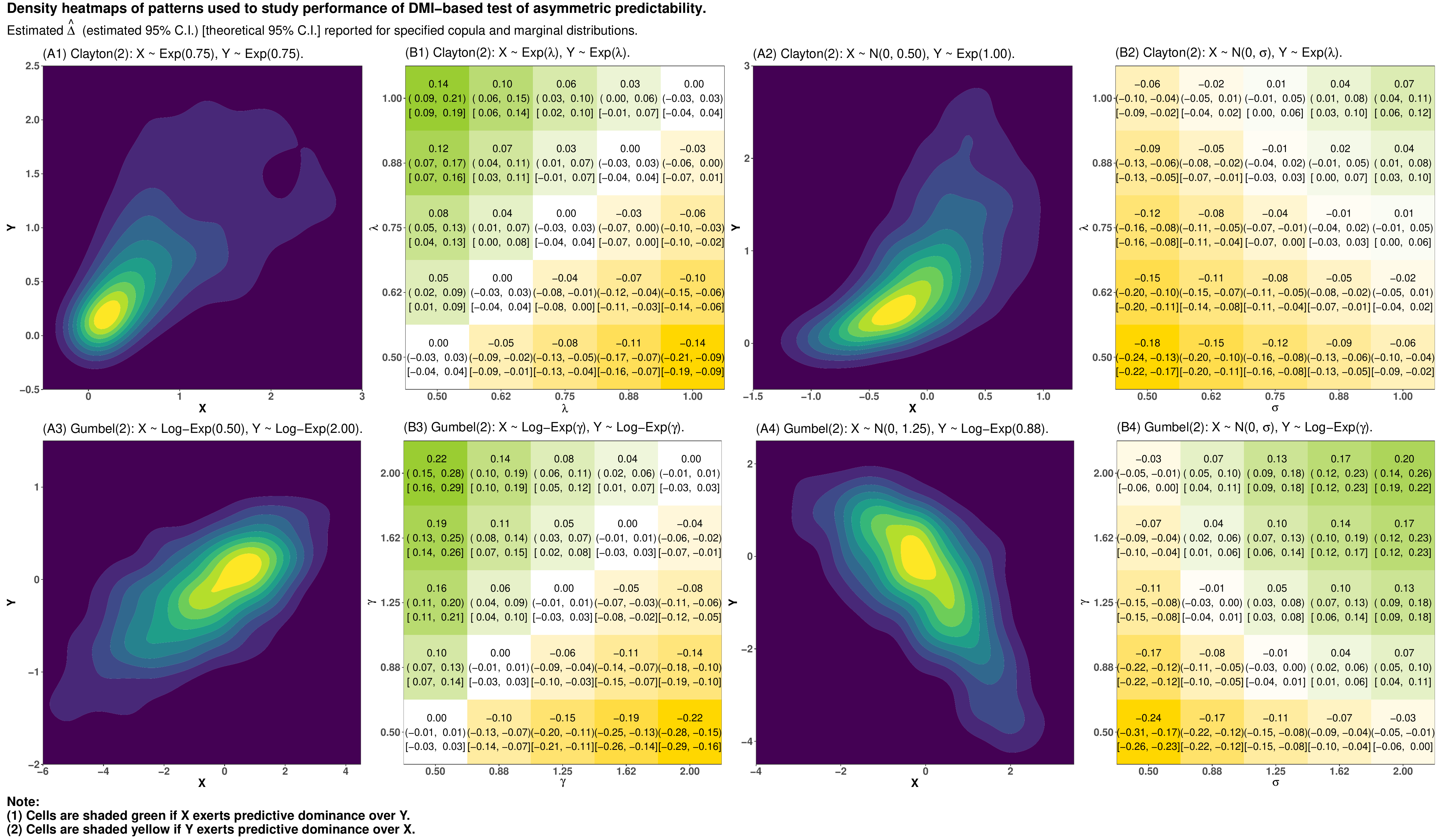}}
\caption{Comparison of theoretical and simulated values of $\hat{\Delta}$ and $95\%$ confidence interval in simulation II.}
\label{fig:simulation2}
\end{figure*}

\section{Data application: DNAm and BP relationship}

Our analysis focuses on a cohort of $525$  children of age 10 - 18 years in the ELEMENT cohort. The aim is to investigate asymmetric association between DNAm  and BP variation.

The NHGRI-EBI GWAS Catalog \citep{Buniello_2018} lists genes in published GWAS for both systolic BP (SBP) and diastolic BP (DBP). In our analysis, we select genes that have been reported to be significantly associated with both SBP and DBP by at least 20 independent studies. Such selection criteria yield three target genes, namely \emph{FGF5}, \emph{ATP2B1} and \emph{PRDM8} for investigation. Each of these genes have a number of CpG sites that are possibly correlated: 21 for \emph{FGF5}, 21 for \emph{ATP2B1} and 51 sites for \emph{PRDM8}. Figure \ref{fig:element} presents an overview of our analysis using $DMI$.

Using our $DMI$-based test, we obtain p - values from testing whether a given CpG site is significantly associated with SBP or DBP. All p - values obtained from the same gene are aggregated at the gene level using the  the Cauchy combination test \citep{Liu_2019}. Our findings reveal that DNAm of gene \emph{FGF5} is associated with SBP (p - value $0.019$), \emph{PRDM8} is associated with SBP ($0.012$) and \emph{ATP2B1} is associated with DBP ($0.042$) in the ELEMENT study at the $5\%$ level of significance, confirming discoveries reported in the GWAS catalog.

Zooming in on  associations at individual CpG methylation sites within each target gene, we apply Bonferroni correction and note that CpG site (i) \emph{CG17564205} within \emph{ATP2B1} is significantly associated with DBP (p - value $2.13 \times 10^{-3}$), (ii) \emph{CG12528713} within \emph{FGF5} is significantly associated with SBP (p - value $1.33 \times 10^{-3}$), and (iii) \emph{CG7462804} within \emph{PRDM8} is significantly associated with SBP (p - value $2.67 \times 10^{-4}$).

For each of the three CpG sites identified above, we further apply the $\Delta$-based test for asymmetry ($\Delta = DMI(\text{DNAm}|\text{BP}) - DMI(\text{BP}|\text{DNAm})$) to examine whether DNAm predictive influence over BP (or converse). We report that two CpG sites (\emph{CG7462804} within \emph{PRDM8} and  \emph{CG12528713} within \emph{FGF5}) do not exhibit predictive asymmetry, with the $\hat{\Delta}\ (95\% \ CI)$ given by $-0.23\ (-4.81, 4.34)$ for \emph{CG7462804} and $-4.76\ (-13.27, 3.74)$ for \emph{CG12528713}. 
Interestingly, DBP is found to exhibit predictive dominance over \emph{CG17564205} within \emph{ATP2B1}, with $\hat{\Delta} = -2.14\ (-3.85, -0.42)$. With most studies on the association between DNAm and blood pressure \citep{Han_2016} still in their infancy, these findings present evidence that DNAm alteration could be influenced by blood pressure and call for further investigation.

\begin{figure*}
\centerline{\includegraphics[width = \textwidth]{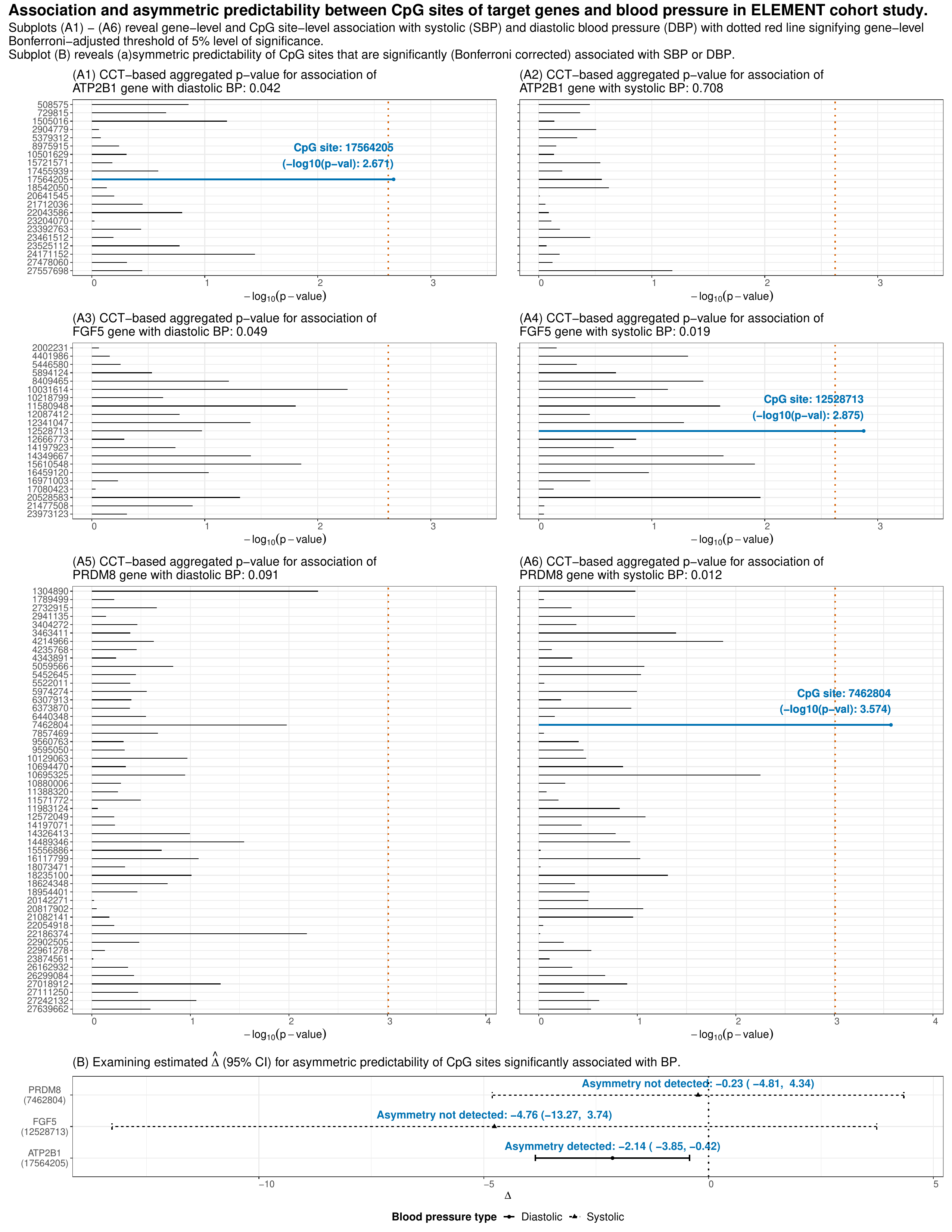}}
\caption{Analysis of ELEMENT \citep{Hernandez_Avila_1996} dataset using $DMI$ and $\Delta$.}
\label{fig:element}
\end{figure*}

\section{Discussion}
Asymmetry is an inherent property of bivariate associations and therefore must not be ignored.  Most dependence measures mask potential asymmetry by implicitly assuming that variables $X$ and $Y$ are equally dependent on each other, which may be false. We present a new causal discovery framework of asymmetric predictability between two random variables $X$ and $Y$ using $DMI$. The $DMI$ is inspired by Shannon's seminal work on information theory and is a well-justified tool that may simultaneously be used to test for association as well as detect and quantify asymmetry, thereby serving as an attractive causal discovery tool. 

A computationally fast and robust Fourier transformation-based method is used to estimate the $DMI$ instead of conventional kernel-based methods. 
Moreover, our method consistently performs faster than existing bandwidth-dependent methods - being approximately 4 orders of magnitude faster for bivariate sample sizes of approximately $10^4$. Another contribution of our $DMI$ methodology is data-splitting inference, that enjoys key large-sample properties necessary for valid inference. This new approach enables us to establish asymptotic normality for functionals of PDFs, which is widely regarded as a difficult issue to address.

Our simulations and data analysis clearly demonstrate the necessity and universal applicability of the quantification of asymmetric predictability in bivariate associations, thereby establishing an attractive causal discovery framework. Potential applications of our  $DMI$ framework include mediation analysis and instrumental variable methods, in which implicit assumptions are made about causal directions, often without justification. In absence of \emph{a priori} knowledge, our framework may serve either as a discovery or confirmatory tool, thereby aiding many applications in current statistical research, particularly in the investigation  of causality. 

\newpage

\label{lastpage}
\backmatter
\section*{Appendix}



%


\subsection{\textbf{Proof of Theorem \ref{thm:strong_conv}}}
\begin{remark}
\label{remark_filter1}
\emph{
The purpose of the filter $\mathbb{I}_{A_{n}}(\vect{t})$ is to define a Fourier-based low-pass filter on the ECF $\mathcal{C}(\vect{t})$ that yields a stable optimal estimate in the minimum \emph{MISE} sense. Primarily, the set $A_{n}$ is specified such that: 
\begin{equation}
\label{eq:filter}
A_{n} = \left\{\vect{t} \in \mathbb{R}^p: \left|\mathcal{C}(\vect{t}) \right|^2 \geq \mathcal{C}_{\text{min}}^2 = \frac{4({n}-1)}{{n}^2} \right\}.
\end{equation} 
Here, the primary filter is necessary for stability of the estimation method because the lower bound $C_{\text{min}}$ can ensure a well-defined square root term in the expression for $\hat{\phi}$. Moreover, according to \cite{Bernacchia2011}, the set $A_{n}$ may exclude an additional small subset of frequencies to produce a smoother density estimate $\hat{f}$. In order for $\hat{f}$ to converge to the true density $f$ as ${n}$ increases, we require that this set of additionally excluded frequencies must shrink, so that the set $A_{n}$ of included frequencies  grows with increasing ${n}$.}
\end{remark}

\begin{remark}
\label{remark_filter2} \emph{ According to \cite{OBrien2016}, the multidimensional ECF $\mathcal{C}(\vect{t})$ consists of a finite set of contiguous hypervolumes denoted by $\left\{HV_l^{n}\right\}_{l = 1}^{k_{n}},$ where $k_{n}$ is a finite integer. Each hypervolume permits ``above-threshold'' frequency values $\vect{t}$ for which the constraint in Equation \ref{eq:filter} holds. Note that at least one such contiguous hypervolume containing $\vect{t} = \vect{0}$ is guaranteed to exist since $\mathcal{C}(\vect{0}) = 1$ due to normalisation and the primary filter $A_{n}$ has a lower bound $\mathcal{C}_{\text{min}} \leq 1$. Following the suggestion by \cite{OBrien2016} we employ the \emph{lowest contiguous hypervolume} filter, choosing the only hypervolume centered at $\vect{t} = \vect{0}$, which we denote as $HV_1^{n}$ for notational convenience. We make the following observations about $HV_1^{n}$:
\begin{enumerate}
    \item The set of frequencies included in the lowest contiguous hypervolume filter are bounded above since they will always be contained within a finite-sized hypervolume around the origin.
    \item The volume of the lowest contiguous hypervolume filter  grows as the number of data points ${n}$ increases, implying more frequencies are included for larger sample sizes.
\end{enumerate}
The resulting filter satisfies the convergence conditions described by \cite{Bernacchia2011}. Hence, we set $A_{n} = HV_1^{n}$, and study convergence of $\hat{f}$ to the true $f$ as ${n}$ increases. For notational convenience, let $\bar{A}_{n}$ denote the complement set of $A_{n}$ and $\mathcal{V}(A_{n})$ denote the volume of $A_{n}$.
}
\end{remark}

\begin{proof}
Note the frequency filter $A_{n}$, its complement $\bar{A}_{n}$ and its volume $\mathcal{V}(A_{n})$ are described in Remarks \ref{remark_filter1} and \ref{remark_filter2}. Since the true density $f$ and the estimator $\hat{f}$ are both square-integrable, we can express them in terms of their corresponding Fourier transforms $\phi$ and $\hat{\phi}$ respectively. Since the characteristic function is integrable, we have, 
\begin{equation*}
    \int \left| \phi(\vect{t}) \right| d\vect{t} < \infty.
\end{equation*} Through the following sequence of inequalities, we are able to establish an upper bound for the absolute error $\left| \hat{f}(\vect{x})-{f}(\vect{x}) \right|$ 
for any $x \in \mathcal{X}$. By definition, note that $\hat{\phi}(\vect{t}) = 0$ for $\vect{t} \notin A_{n}$. To establish  Theorem \ref{thm:strong_conv}, it is sufficient to show that the upper bound of the absolute error given below tends to zero as ${n} \rightarrow \infty$. We have:
\begin{align}
    \label{eq:pushright2}
        &\left| \hat{f}(\vect{x})-{f}(\vect{x}) \right| \notag\\
        &= \left|\frac{1}{(2 \pi)^p} \int_{\mathbb{R}^p} \exp (-\mathrm{i} \vect{t}^\prime \vect{x})\left\{\hat{\phi}(\vect{t})-\phi(\vect{t})\right\} \mathrm{d} \vect{t}\right| \notag\\
        &\leq \frac{1}{(2 \pi)^p} \int_{\mathbb{R}^p} \lvert \exp (-\mathrm{i} \vect{t}^\prime \vect{x}) \rvert \lvert  \hat{\phi}(\vect{t})-\phi(\vect{t}) \rvert \mathrm{d} \vect{t} \notag\\
        &= \frac{1}{(2 \pi)^p} \int_{\mathbb{R}^p} \lvert  \hat{\phi}(\vect{t})-\phi(\vect{t}) \rvert \mathrm{d} \vect{t} \notag\\
        &= \frac{1}{(2 \pi)^p} \int_{A_{n}}  \lvert  \hat{\phi}(\vect{t})-\phi(\vect{t}) \rvert \mathrm{d} \vect{t} + 
         \frac{1}{(2 \pi)^p} \int_{\bar{A}_{n}}  \lvert  \phi(\vect{t}) \rvert \mathrm{d} \vect{t} \notag\\
         &\leq \frac{1}{(2 \pi)^p} \int_{A_{n}}  \lvert  \hat{\phi}(\vect{t})- \mathcal{C}(\vect{t}) \rvert \mathrm{d} \vect{t} + 
         \frac{1}{(2 \pi)^p} \int_{A_{n}}  \lvert \mathcal{C}(\vect{t}) -  \phi(\vect{t}) \rvert \mathrm{d} \vect{t}  \notag\\
         &  \quad \quad + \frac{1}{(2 \pi)^p} \int_{\bar{A}_{n}}  \lvert  \phi(\vect{t}) \rvert \mathrm{d} \vect{t} \notag\\
         &:= D_1 + D_2 + D_3.
\end{align} Under the assumptions, $\lim_{{n} \rightarrow \infty} \ \mathcal{V}\left(A_{n}\right) =  \infty$ and $\lim_{{n} \rightarrow \infty} \ \mathcal{V}\left(\bar{A}_{n}\right) =  0$. Consequently, the second term in Equation \ref{eq:pushright2}, $D_2 \rightarrow 0$ as ${n} \rightarrow \infty$ due to Theorem 1 of \cite{Csorgo1983}. Further, $D_3 \leq \mathcal{V}\left(\bar{A}_{n}\right)/(2\pi)^p$, since $\left|\phi(\vect{t}) \right| \leq 1$. Consequently,  $D_3 \rightarrow 0$ as ${n} \rightarrow \infty$.

\noindent To prove $D_1 \rightarrow 0$ as ${n} \rightarrow \infty$, we first consider the two following disjoint sets, 
\begin{equation*}
    \begin{aligned}
        B^{+}_{n} = \{\vect{t}: \lvert \mathcal{C}(\vect{t}) \rvert^2 \geq 4({n}-1)/{n}^2\}, \\
        B^{-}_{n} = \{\vect{t}: \lvert \mathcal{C}(\vect{t}) \rvert^2 < 4({n}-1)/{n}^2\}. \\
    \end{aligned}
\end{equation*} We rewrite the first integral $D_1$ as follows
\footnotesize{\begin{equation*}
    \begin{aligned}
    D_1  &= \frac{1}{2 \pi} \int_{A_{n} \cap B^{+}_{n}} \lvert \mathcal{C}(\vect{t}) \rvert 
        \left(1-\frac{{n}}{2({n}-1)}\left[1+\sqrt{1-\frac{4({n}-1)}{ \lvert{n} \mathcal{C}(\vect{t})\rvert^{2}}}\right]  \right) d \vect{t} \notag\\ & \quad + \frac{1}{2 \pi} \int_{A_{n} \cap B^{-}_{n}} \lvert \mathcal{C}(\vect{t}) \rvert d \vect{t} \\ 
        &:= D_4 + D_5.
        \end{aligned}
\end{equation*}}
\normalsize The first term $D_4$ may be simplified by noting that for $\vect{t} \in B^{+}_{n}$, we have $\lvert \mathcal{C}(\vect{t}) \rvert^2 \geq 4({n}-1)/{n}^2$. This ensures a non-negative argument under the square root operation. Using the inequality $\sqrt{1-x} + \sqrt{x} \geq 1$ for $0 \leq x \leq 1$ for $D_4$, and using the inequality $$\lvert \mathcal{C}(\vect{t}) \rvert \leq \sqrt{4({n}-1)/{n}^2} \text{ for } t \in B_{n}^{-},$$ we establish that $D_1$ is bounded as follows:
\begin{align}
\label{eq:pushright3}
       D_1 &= D_4 + D_5 \notag\\
       & \leq \frac{1}{(2 \pi)^p} \int_{A_{n} \cap B^{+}_{n}}  \left\{ \frac{1}{\sqrt{{n}-1}} - \frac{\lvert \mathcal{C}(\vect{t}) \rvert}{{n}-1} \right\}   d\vect{t}   \notag\\
        & \quad + 
        \frac{\sqrt{4({n}-1)}}{{n}(2\pi)^p } \int_{A_{n} \cap B^{-}_{n}} d \vect{t} \notag\\
        &\leq  \frac{1}{(2 \pi)^p}  \left\{ \frac{1}{\sqrt{{n}-1}} + \frac{1}{{n}-1} \right\} \int_{A_{n} \cap B^{+}_{n}} d\vect{t}   \notag\\
        & \quad + 
        \frac{\sqrt{4({n}-1)}}{{n} (2\pi)^p} \int_{A_{n}} d \vect{t} \notag\\
        &= \frac{1}{(2 \pi)^p}  \left\{ \frac{1}{\sqrt{{n}-1}} + \frac{1}{{n}-1} \right\} \mathcal{V}(A_{n} \cap B_{n}^+) \notag\\
        & \quad +  \frac{\sqrt{4({n}-1)}}{{n} (2\pi)^p} \mathcal{V}(A_{n} \cap B_{n}^-) \notag\\
        & \leq \frac{1}{(2 \pi)^p} \left\{ \frac{1}{\sqrt{{n}-1}} + \frac{1}{{n}-1} + \frac{\sqrt{4({n}-1)}}{{n}} \right\} \mathcal{V}(A_{n}) \notag\\
        & \leq \frac{1}{(2 \pi)^p} \left\{ \frac{1}{\sqrt{{n}-1}} + \frac{1}{{n}-1} + \frac{2}{\sqrt{{n}-1}} \right\}  \mathcal{V}(A_{n}).
        \end{align} 
The assumptions in Theorem \ref{thm:strong_conv} include $\mathcal{V}(A_{n})/\sqrt{{n}} \rightarrow 0$ as $n_1 \rightarrow \infty$, which ensures that the upper bound in  Equation \ref{eq:pushright3} tends to zero for large ${n}$. In summary, assuming  $\mathcal{V}(A_{n}) \rightarrow \infty$, $\mathcal{V}(\bar{A}_{n}) \rightarrow 0$, and $\mathcal{V}(A_{n})/\sqrt{{n}} \rightarrow 0$ as $n_1 \rightarrow \infty$, we have $\left| \hat{f}(\vect{x}) -  {f}(\vect{x})\right| \rightarrow 0$ for every $\vect{x} \in \mathcal{X}$.
\end{proof}

\subsection{\textbf{Proof of Theorem \ref{thm:strong_conv_mi}}}
\begin{proof}
From Theorem \ref{thm:strong_conv}, for any small $\epsilon$, there exists sufficiently large $n$ such that $$\left|\hat{c}(\vect{z}) - c(\vect{z}) \right| < \epsilon \text{ for all }\vect{z} \in [0, 1]^2.$$ We define $\delta_j = \hat{c}(\vect{Z}_j) - {c}(\vect{Z}_j)$ and consider the Taylor expansion for $\log\left(1 + \delta_j/c(\vect{Z}_j) \right)$ as follows
\begin{equation*}
\log(\hat{c}(\vect{Z}_j)) = \log({c}(\vect{Z}_j)) + \frac{\delta_j}{{c}(\vect{Z}_j)} + o(\epsilon).
\end{equation*}
Note that $\left|\hat{c}(\vect{Z}_j) - c(\vect{Z}_j) \right| < \epsilon$ and the term $o(\epsilon)$ may be ignored. Further, note that $c$ is bounded below by a constant $B^{-1}$ on the support of $c$, i.e., 
\begin{equation*}
    {c}\left(\vect{z}\right) > B^{-1}, \quad \vect{z} \in D_c.
\end{equation*} Consequently, we have 
\begin{align*}
    \left|\frac{1}{n} \sum_{j=1}^{n} \log \hat{c}\left(\vect{Z}_j\right)-\frac{1}{n} \sum_{j=1}^{n} \log {c}\left(\vect{Z}_j\right)\right| &\leq \frac{1}{n} \sum_{j=1}^{n}\left|\frac{\delta_{j}}{{c}\left(\vect{Z}_j\right)}\right| \notag\\
    & \leq \frac{ \epsilon}{n} \sum_{j=1}^{n} \frac{1}{\left|{c}\left(\vect{Z}_j\right)\right|} \notag\\
    & \leq \frac{\epsilon}{B}, 
\end{align*} which implies 
\begin{align*}
    \left|\frac{1}{n} \sum_{j=1}^{n} \log \hat{c}\left(\vect{Z}_j\right)-\frac{1}{n} \sum_{j=1}^{n} \log {c}\left(\vect{Z}_j\right)\right| \overset{a.s.}{\rightarrow} 0, \text{ as } n \rightarrow \infty, \text{ on } D_c. 
\end{align*} Thus, as $n \rightarrow \infty$, we obtain $\widehat{MI} - MI_0 \overset{a.s.}{\rightarrow} 0$ on $D_c$. Further, by the strong law of large numbers, as $n \rightarrow \infty$, we have $MI_0 -  MI\overset{a.s.}{\rightarrow} 0$  on $\mathbb{R}^p$; consequently as $n \rightarrow \infty$, we have $\widehat{MI} - MI  \overset{a.s.}{\rightarrow} 0$ on $D_c$. By noting that a similar result holds for $ER$, and using the continuous mapping theorem,  we are able to show $\widehat{DMI} - DMI  \overset{a.s.}{\rightarrow} 0$ . 
This concludes the proof.
\end{proof}

\subsection{\textbf{Proof of Theorems \ref{thm:dmi_normality} and \ref{thm:del_normality}}}

In this section we analyze the large-sample behaviour of $\widehat{DMI}$ and $\hat{\Delta}$ using a data-splitting technique. This technique splits the available data $\mathcal{D}$, of size $n$, into two disjoint sets, resulting in  $\mathcal{D}_1$ and $\mathcal{D}_2$ with sample sizes $n_1$ and $n_2$ respectively ($n_1 + n_2 = n$). We require that both $n_1, n_2 \rightarrow \infty$ as $n \rightarrow \infty$, or $n_1 \wedge n_2 := \min(n_1, n_2) \rightarrow \infty$. Data from $\mathcal{D}_1$ is used to obtain estimates of the copula density function $\hat{c}_{\mathcal{D}_1}$ and the marginal density functions $\hat{f}_{X; \mathcal{D}_1}$ and  $\hat{f}_{Y; \mathcal{D}_1}$. Under some mild assumptions, these density estimators converge uniformly to their population counterparts on their respective support sets as sample size $n_1 \rightarrow \infty$. See Theorem \ref{thm:strong_conv} for more details. 

Using data from $\mathcal{D}_2$ we evaluate the estimators $\widehat{MI}, \hat{H}(X)$ and $\hat{H}(Y)$ as follows:
\begin{gather}
   \begin{aligned}
\label{eq:dmi_parts_estim_2}
    \widehat{MI} &= {(n_2)}^{-1}\sum_{j=1}^{n_2}  \log \left\{\hat{c}_{\mathcal{D}_1}(\vect{Z}^{D_2}_j)\right\},\\
    \hat{H}(X) &= {(n_2)}^{-1}\sum_{j=1}^{n_2}  - \log \left\{\hat{f}_{X; \mathcal{D}_1}(X^{D_2}_j)\right\},\\
    \hat{H}(Y) &= {(n_2)}^{-1}\sum_{j=1}^{n_2}  - \log \left\{\hat{f}_{Y; \mathcal{D}_1}(Y^{D_2}_j)\right\}, 
\end{aligned}  
\end{gather} which are used to estimate $\widehat{ER}(X|Y)$ and $\widehat{ER}(Y|X)$. This yields $\widehat{DMI}(X|Y) = \widehat{MI}(X, Y)  \times  \widehat{ER}(X|Y)$,  $\widehat{DMI}(Y|X) = \widehat{MI}(X, Y)  \times  \widehat{ER}(Y|X)$ and $\hat{\Delta} = \widehat{DMI}(X|Y) - \widehat{DMI}(Y|X)$. Likewise, we define the oracle estimators based on the true density functions as follows:
\begin{gather}
   \begin{aligned}
\label{eq:dmi_parts_estim_oracle}
    {MI}_0 &= {(n_2)}^{-1}\sum_{j=1}^{n_2}  \log \left\{{c}(\vect{Z}^{D_2}_j)\right\},\\
    {H}_0(X) &= {(n_2)}^{-1}\sum_{j=1}^{n_2}  - \log \left\{{f}_{X}(X^{D_2}_j)\right\},\\
    {H}_0(Y) &= {(n_2)}^{-1}\sum_{j=1}^{n_2}  - \log \left\{{f}_{Y}(Y^{D_2}_j)\right\}.
\end{aligned}  
\end{gather}

\begin{lemma}
\label{thm:slutsky1}
Assuming the conditions presented in Theorem \ref{thm:dmi_normality} hold, we have
\begin{equation}
\label{eq:thm2_eq1}
    \sqrt{n_2}
\begin{pmatrix}
\widehat{MI} - MI_0 \\
\hat{H}(X) - H_0(X) \\
\hat{H}(Y) - H_0(Y)
\end{pmatrix}\overset{\mathcal{P}}{\rightarrow}  \vect{0}, \text{ as } n_1 \wedge n_2 \rightarrow \infty.
\end{equation}
\end{lemma}

\begin{proof}
Note that 
\begin{align*}
   \sqrt{n_2}\left(\widehat{MI} - MI_0 \right) &=  \frac{1}{\sqrt{n_2}} \sum_{j = 1}^{n_2} \frac{\hat{c}_{\mathcal{D}_1}(\vect{Z}^{D_2}_j) - {c}(\vect{Z}^{D_2}_j)}{{c}(\vect{Z}^{D_2}_j)} \notag\\ \quad & \quad + o_{\mathcal{P}}(1/\sqrt{n_2}), \\
   &= \frac{S_{n_2}}{\sqrt{n_2}} + o_{\mathcal{P}}(1/\sqrt{n_2}).
\end{align*} We now show that the leading term $S_{n_2}/\sqrt{n_2} \overset{\mathcal{P}}{\rightarrow} 0$ as $n_1 \wedge n_2 \rightarrow 0$, which will establish $\sqrt{n_2}\left(\widehat{MI} - MI_0 \right)  \overset{\mathcal{P}}{\rightarrow} 0$ as $n_1 \wedge n_2 \rightarrow 0$. Using similar arguments we may establish $\sqrt{n_2}\left(\hat{H}(X) - H_0(X) \right)  \overset{\mathcal{P}}{\rightarrow} 0$ as well as $\sqrt{n_2}\left(\hat{H}(Y) - H_0(Y) \right)  \overset{\mathcal{P}}{\rightarrow} 0$ as  $n_1 \wedge n_2 \rightarrow 0$. In conjunction with asymptotic normality of the oracle estimators $MI_0, H_0(X)$ and $H_0(Y)$ and Equation \ref{eq:thm2_eq1}, we prove Theorems \ref{thm:dmi_normality} and \ref{thm:del_normality}. 

It is sufficient to show $\mathbb{E}\left\{ \left(S_{n_2}/\sqrt{n_2} \right)^2\right\} \rightarrow 0$ as $n_1, n_2 \rightarrow \infty$. 
Note that  $\mathbb{E}\left\{ \left(S_{n_2}/\sqrt{n_2} \right)^2\right\} = \mathbb{E}^2 \left(S_{n_2}/\sqrt{n_2} \right) + \mathbb{V} \left(S_{n_2}/\sqrt{n_2} \right)$.  First, we prove $\mathbb{E}\left( S_{n_2}/\sqrt{n_2} \right) = 0$. This is because
\footnotesize{\begin{align*}
& \mathbb{E}\left(\frac{S_{n_2}}{\sqrt{{n_2}}} \right) = \mathbb{E}_{\mathcal{D}_1} \left\{ \mathbb{E}_{\vect{Z} \left. \right| \mathcal{D}_1}\left( \left. \frac{S_{n_2}}{\sqrt{{n_2}}}   \right| \mathcal{D}_1 \right) \right\} \\
&= \mathbb{E}_{\mathcal{D}_1} \left\{ \mathbb{E}_{\vect{Z} \left. \right| \mathcal{D}_1}\left( \frac{1}{\sqrt{{n_2}}} \left.\sum_{j = 1}^{n_2} \frac{\hat{c}_{\mathcal{D}_1}(\vect{Z}^{D_2}_j) - {c}(\vect{Z}^{D_2}_j)}{{c}(\vect{Z}^{D_2}_j)} \right| \mathcal{D}_1 \right) \right\} \\
&= \mathbb{E}_{\mathcal{D}_1} \left\{  {\sqrt{{n_2}}}\  \mathbb{E}_{\vect{Z} \left. \right| \mathcal{D}_1}\left(  \left. \frac{\hat{c}_{\mathcal{D}_1}(\vect{Z}) - {c}(\vect{Z})}{{c}(\vect{Z})} \right| \mathcal{D}_1\right) \right\}, 
\end{align*}} \normalsize where the inner expectation term is evaluated as follows
\begin{align*}
 & \mathbb{E}_{\vect{Z} \left. \right| \mathcal{D}_1}\left(  \left. \frac{\hat{c}_{\mathcal{D}_1}(\vect{Z}) - {c}(\vect{Z})}{{c}(\vect{Z})} \right| \mathcal{D}_1\right) \\
 &= \int_{\vect{z}} \left( \frac{\hat{c}_{\mathcal{D}_1}(\vect{z}) - {c}(\vect{z})}{{c}(\vect{z})} \right) {c}(\vect{z}) d\vect{z} \\
&= \int_{\vect{z}} \hat{c}_{\mathcal{D}_1}(\vect{z}) d\vect{z}  - \int_{\vect{z}} {c}(\vect{z}) d\vect{z} \\
&= \hat{\phi}_{\mathcal{D}_1}(\vect{0}) - 1 = 0.
\end{align*} The last equality holds since $\hat{\phi}_{\mathcal{D}_1}$ is the Fourier transform associated with the optimal density function estimator $\hat{c}_{\mathcal{D}_1}$, and we know $\hat{\phi}_{\mathcal{D}_1}(\vect{0}) = 1$ and consequently, $\mathbb{E}\left( S_{n_2}/\sqrt{{n_2}} \right) = 0$. Next, we consider the term  $\mathbb{V}\left(S_{n_2}/\sqrt{{n_2}} \right)$:
\begin{align}
\label{eq:thm2_eq2}
\mathbb{V}\left(\frac{S_{{n_2}}}{\sqrt{{n_2}}}\right) &= \mathbb{E}_{\mathcal{D}_1} \left\{ \mathbb{V}_{\vect{Z} \left. \right| \mathcal{D}_1}\left( \left. \frac{S_{n_2}}{\sqrt{{n_2}}}   \right| \mathcal{D}_1 \right) \right\} \notag\\ & \quad + \mathbb{V}_{\mathcal{D}_1} \left\{ \mathbb{E}_{\vect{Z} \left. \right| \mathcal{D}_1}\left( \left. \frac{S_{n_2}}{\sqrt{{n_2}}}   \right| \mathcal{D}_1 \right) \right\},
\end{align} where the second term is already shown to be zero. Note that, conditional on $\mathcal{D}_1$, the terms $\hat{c}(\vect{Z}^{\mathcal{D}_2}_j)$ are independent and identically distributed for all $\vect{Z}^{\mathcal{D}_2}_j \in \mathcal{D}_2$. We have
\begin{align*}
 & \mathbb{V}_{\vect{Z} \left. \right| \mathcal{D}_1}\left( \left. \frac{S_{n_2}}{\sqrt{{n_2}}}   \right| \mathcal{D}_1 \right) \\
 & = \mathbb{V}_{\vect{Z} \left. \right| \mathcal{D}_1}\left (\frac{1}{\sqrt{{n_2}}} \left.\sum_{j = 1}^{n_2} \frac{\hat{c}_{\mathcal{D}_1}(\vect{Z}^{D_2}_j) - {c}(\vect{Z}^{D_2}_j)}{{c}(\vect{Z}^{D_2}_j)} \right| \mathcal{D}_1 \right) \\
 &=  \mathbb{V}_{\vect{Z} \left. \right| \mathcal{D}_1 }\left(  \left. \frac{\hat{c}_{\mathcal{D}_1}(\vect{Z}) - {c}(\vect{Z})}{{c}(\vect{Z})} \right| \mathcal{D}_1\right)\\
 &= \mathbb{E}_{\vect{Z} \left. \right| \mathcal{D}_1} \left\{ \left(  \left. \frac{\hat{c}_{\mathcal{D}_1}(\vect{Z}) - {c}(\vect{Z})}{{c}(\vect{Z})}  \right)^2 \right| \mathcal{D}_1\right\} , 
\end{align*} since $ \mathbb{E}_{\vect{Z} \left. \right| \mathcal{D}_1}\left(  \left. \frac{\hat{c}_{\mathcal{D}_1}(\vect{Z}) - {c}(\vect{Z})}{{c}(\vect{Z})} \right| \mathcal{D}_1\right) = 0$. Moreover, 
\begin{align*}
& \mathbb{E}_{\vect{Z} \left. \right| \mathcal{D}_1} \left\{ \left(  \left. \frac{\hat{c}_{\mathcal{D}_1}(\vect{Z}) - {c}(\vect{Z})}{{c}(\vect{Z})}  \right)^2 \right| \mathcal{D}_1\right\} \\
& = \int_{\vect{z}} \left(  \frac{\hat{c}_{\mathcal{D}_1}(\vect{z}) - {c}(\vect{z})}{{c}(\vect{z})}\right)^2 {c}(\vect{z}) d\vect{z} \\
&\leq B \int_{\vect{z}} \left(\hat{c}_{\mathcal{D}_1}(\vect{z}) - {c}(\vect{z}) \right)^2 d\vect{z}. 
\end{align*} where $B$ is  a (positive) lower bound for the density $c$ over its support. Plugging this inequality into Equation \ref{eq:thm2_eq2}, we get 
\begin{align*}
\mathbb{V}\left(\frac{S_{{n_2}}}{\sqrt{{n_2}}}\right) &\leq B \times  \mathbb{E}_{\mathcal{D}_1} \left\{  \int_{\vect{z}} \left(\hat{c}_{\mathcal{D}_1}(\vect{z}) - {c}(\vect{z}) \right)^2 d\vect{z} \right\},\\
&= B \times MISE(\hat{c}, c).
\end{align*} \cite{Bernacchia2011} present an expression of $MISE$ in terms of the optimal kernel and  prove that the last expression goes to zero as sample size increases, i.e., $ MISE(\hat{c}, c) \rightarrow 0$ as $n_1 \wedge n_2 \rightarrow \infty$. This allows us to claim $\sqrt{n_2} \left(\widehat{MI} - MI_0 \right) \overset{\mathcal{P}}{\rightarrow} 0$  as $n_1 \wedge n_2 \rightarrow \infty$. 
Note that the arguments presented above are generally valid for any true density function that is bounded away from zero and infinity on its support. Hence, they can also be used to establish similar results involving $\hat{H}(X)$ and $\hat{H}(Y)$ as well, thereby concluding the proof.
\end{proof}
\begin{lemma}
\label{thm:slutsky2}
By the multivariate central limit theorem we have for $MI \neq 0$, the oracle estimators jointly converge in distribution to a three-dimensional normal distribution, namely 
\begin{equation}
\label{eq:thm2_eq1}
    \sqrt{n_2}
\begin{pmatrix}
MI_0 - MI \\
H_0(X) - H(X) \\
H_0(Y) - H(Y)
\end{pmatrix}\overset{\mathcal{D}}{\rightarrow}  N\left(\vect{0}, \Sigma\right), \text{ as } n_1 \wedge n_2 \rightarrow \infty,
\end{equation} where $\Sigma$ is a $3 \times 3$ diagonal matrix. 
\end{lemma}
\begin{remark}
\label{rem:avar}
Note how $MI_0$ involves the log-copula density $c$ alone, whereas $H_0(X)$ and $H_0(Y)$ involve the log-marginal densities $f_X$ and $f_Y$ respectively. Clearly, $MI_0$ explains shared or joint behavior while $H_0(X)$ and $H_0(Y)$ explain marginal behaviour. These three quantities have no shared population attributes, implying that $\Sigma $ is a diagonal matrix $\text{diag}\left(\sigma_1^2, \sigma_2^2, \sigma_3^2 \right)$, where 
\begin{align*}
        \sigma_1^2 &= \mathbb{V}\left[\log \left\{{c}(\vect{Z})\right\} \right],\\
        \sigma_2^2 &= \mathbb{V}\left[\log \left\{{f}_X(X)\right\} \right],\\
        \sigma_3^2 &= \mathbb{V}\left[\log \left\{{f}_Y(Y)\right\} \right].
\end{align*} These variance terms may be estimated using standard Monte-Carlo methods. 
\end{remark} 

Lemma \ref{thm:slutsky1} and \ref{thm:slutsky2}, in conjuction with Slutsky's theorem \citep{billingsley} allows us to claim 
\begin{gather}
\begin{aligned}
\label{eq:thm3_eq1}
& \sqrt{n_2}
\begin{pmatrix}
\widehat{MI} - MI \\
\hat{H}(X) - H(X) \\
\hat{H}(Y) - H(Y)
\end{pmatrix} \\
& =   
\sqrt{n_2}
\begin{pmatrix}
MI_0 - MI \\
H_0(X) - H(X) \\
H_0(Y) - H(Y)
\end{pmatrix}  +
\sqrt{n_2}
\begin{pmatrix}
\widehat{MI} - MI_0 \\
\hat{H}(X) - H_0(X) \\
\hat{H}(Y) - H_0(Y)
\end{pmatrix} \\
& \overset{\mathcal{D}}{\rightarrow}  N\left(\vect{0}, \Sigma\right) + o_{\mathcal{P}}(1), 
\end{aligned}
\end{gather} where $\Sigma$ is the asymptotic dispersion matrix as described in Remark \ref{rem:avar}.  

We are now in a position to prove Theorems \ref{thm:dmi_normality} and \ref{thm:del_normality}.
The proofs are very closely related and we only present the proof for Theorem \ref{thm:del_normality}. 
\begin{proof}
Considering the smooth function $$g(a, b, c) = a \times \frac{\exp(b) - \exp(c)}{\exp(b) + \exp(c)},$$ we write $\hat{\Delta} = g(\widehat{MI}, \hat{H}(X), \hat{H}(Y))$ and ${\Delta} = g({MI}, {H}(X), {H}(Y))$. Using the following decomposition
\begin{equation*}
\begin{aligned}
H(X, Y) &= H(X|Y) + H(Y)\\
&= H(Y|X) + H(X),
\end{aligned}
\end{equation*} it follows that
\begin{equation*}
\begin{aligned}
\Delta &= MI \times \frac{\exp(H(X|Y)) - \exp(H(Y|X)))}{\exp(H(X|Y)) + \exp(H(Y|X)))} \\
&= MI \times \frac{\exp(H(X)) - \exp(H(Y)))}{\exp(H(X)) + \exp(H(Y)))}.
\end{aligned}
\end{equation*}  We may rewrite $\Delta$ as $g(MI, H(X), H(Y))$. Using the multivariate delta method \citep{billingsley}, we get 
\begin{align*}
    \sqrt{n_2} \left(g(\widehat{MI}, \hat{H}(X), \hat{H}(Y)) - g({MI}, {H}(X), {H}(Y)) \right) \\ \overset{\mathcal{D}} {\rightarrow}  N\left(0, \sigma^2\right), \text{ as } n_1 \wedge n_2 \rightarrow \infty, 
\end{align*} $\sigma^2 = \left\{ \nabla g\left(MI, H(X), H(Y)\right) \right\}^T \Sigma \left\{ \nabla g \left(MI, H(X), H(Y)\right) \right\}$ and $\nabla g$ denotes the vector of gradients for $g$ with respect to its arguments. A little algebra yields 
\begin{align*}
    \sigma^2 &= \left\{2ER(X|Y) - 1\right\}^2 \sigma_1^2 \\
    & \quad + 4 \left[ DMI(X|Y) \left\{1 - ER(X|Y)\right\} \right]^2 \left(\sigma_2^2 + \sigma_3^2\right). 
\end{align*} This concludes the proof of Theorem \ref{thm:del_normality}. \end{proof}


\bibliographystyle{biom} 
\bibliography{references}

\begin{thebibliography}{}

\bibitem[\protect\citeauthoryear{Bernacchia and Pigolotti}{Bernacchia and
  Pigolotti}{2011}]{Bernacchia2011}
Bernacchia, A. and Pigolotti, S. (2011).
\newblock Self-consistent method for density estimation.
\newblock {\em Journal of the Royal Statistical Society: Series B (Statistical
  Methodology)} {\bf 73,} 407--422.

\bibitem[\protect\citeauthoryear{Billingsley}{Billingsley}{1995}]{billingsley}
Billingsley, P. (1995).
\newblock {\em Probability and measure}.
\newblock A Wiley-Interscience publication. Wiley, New York [u.a.], 3. ed
  edition.

\bibitem[\protect\citeauthoryear{Buniello, MacArthur, Cerezo, Harris, Hayhurst,
  Malangone, McMahon, Morales, Mountjoy, Sollis, Suveges, Vrousgou, Whetzel,
  Amode, Guillen, Riat, Trevanion, Hall, Junkins, Flicek, Burdett, Hindorff,
  Cunningham, and Parkinson}{Buniello et~al.}{2018}]{Buniello_2018}
Buniello, A., MacArthur, J. A.~L., Cerezo, M., Harris, L.~W., Hayhurst, J.,
  Malangone, C., McMahon, A., Morales, J., Mountjoy, E., Sollis, E., Suveges,
  D., Vrousgou, O., Whetzel, P.~L., Amode, R., Guillen, J.~A., Riat, H.~S.,
  Trevanion, S.~J., Hall, P., Junkins, H., Flicek, P., Burdett, T., Hindorff,
  L.~A., Cunningham, F., and Parkinson, H. (2018).
\newblock The {NHGRI}-{EBI} {GWAS} catalog of published genome-wide association
  studies, targeted arrays and summary statistics 2019.
\newblock {\em Nucleic Acids Research} {\bf 47,} D1005--D1012.

\bibitem[\protect\citeauthoryear{Cs\"{o}rg\H{o} and Totik}{Cs\"{o}rg\H{o} and
  Totik}{1983}]{Csorgo1983}
Cs\"{o}rg\H{o}, S. and Totik, V. (1983).
\newblock On how long interval is the empirical characteristic function
  uniformly consistent?
\newblock {\em Acta Sci. Math. (Szeged)} {\bf 45,} 141--149.

\bibitem[\protect\citeauthoryear{Czado}{Czado}{2019}]{czado_2019}
Czado, C. (2019).
\newblock {\em Analyzing Dependent Data with Vine Copulas: A Practical Guide
  With R (Lecture Notes in Statistics, 222)}.
\newblock Springer, paperback edition.

\bibitem[\protect\citeauthoryear{Dicorpo, Lent, Guan, Hivert, and
  Pankow}{Dicorpo et~al.}{2018}]{dicorpo_2018}
Dicorpo, D.~A., Lent, S., Guan, W., Hivert, M.-F., and Pankow, J.~S. (2018).
\newblock Mendelian randomization suggests causal influence of glycemic traits
  on {DNA} methylation.
\newblock {\em Diabetes} {\bf 67,}.

\bibitem[\protect\citeauthoryear{Duong and Hazelton}{Duong and
  Hazelton}{2005}]{Duong2005}
Duong, T. and Hazelton, M.~L. (2005).
\newblock Cross-validation bandwidth matrices for multivariate kernel density
  estimation.
\newblock {\em Scandinavian Journal of Statistics} {\bf 32,} 485--506.

\bibitem[\protect\citeauthoryear{Han, Liu, Duan, Perry, Li, and He}{Han
  et~al.}{2016}]{Han_2016}
Han, L., Liu, Y., Duan, S., Perry, B., Li, W., and He, Y. (2016).
\newblock {DNA} methylation and hypertension: emerging evidence and challenges.
\newblock {\em Briefings in Functional Genomics} page elw014.

\bibitem[\protect\citeauthoryear{Heidenreich, Schindler, and
  Sperlich}{Heidenreich et~al.}{2013}]{Heidenreich2013}
Heidenreich, N.-B., Schindler, A., and Sperlich, S. (2013).
\newblock Bandwidth selection for kernel density estimation: a review of fully
  automatic selectors.
\newblock {\em {AStA} Advances in Statistical Analysis} {\bf 97,} 403--433.

\bibitem[\protect\citeauthoryear{Hernandez-Avila, Gonzalez-Cossio, Palazuelos,
  Romieu, Aro, Fishbein, Peterson, and Hu}{Hernandez-Avila
  et~al.}{1996}]{Hernandez_Avila_1996}
Hernandez-Avila, M., Gonzalez-Cossio, T., Palazuelos, E., Romieu, I., Aro, A.,
  Fishbein, E., Peterson, K.~E., and Hu, H. (1996).
\newblock Dietary and environmental determinants of blood and bone lead levels
  in lactating postpartum women living in mexico city.
\newblock {\em Environmental Health Perspectives} {\bf 104,} 1076--1082.

\bibitem[\protect\citeauthoryear{Joe}{Joe}{2014}]{joe_2014}
Joe, H. (2014).
\newblock {\em Dependence Modeling with Copulas (Chapman \& Hall/CRC Monographs
  on Statistics and Applied Probability)}.
\newblock Chapman and Hall/CRC, hardcover edition.

\bibitem[\protect\citeauthoryear{Kreyszig}{Kreyszig}{2020}]{Kreyszig_Erwin2020-07-21}
Kreyszig, E. (2020).
\newblock {\em Advanced Engineering Mathematics}.
\newblock Wiley, loose leaf edition.

\bibitem[\protect\citeauthoryear{Liu and Xie}{Liu and Xie}{2019}]{Liu_2019}
Liu, Y. and Xie, J. (2019).
\newblock Cauchy combination test: A powerful test with analytic p-value
  calculation under arbitrary dependency structures.
\newblock {\em Journal of the American Statistical Association} {\bf 115,}
  393--402.

\bibitem[\protect\citeauthoryear{Ma and Sun}{Ma and Sun}{2008}]{Ma2011}
Ma, J. and Sun, Z. (2008).
\newblock Mutual information is copula entropy.
\newblock {\em Tsinghua Science and Technology} {\bf 16,} 51--54.

\bibitem[\protect\citeauthoryear{Nelsen}{Nelsen}{2006}]{nelsen_2006}
Nelsen, R.~B. (2006).
\newblock {\em An Introduction to Copulas}.
\newblock Springer New York.

\bibitem[\protect\citeauthoryear{O'Brien, Kashinath, Cavanaugh, Collins, and
  O'Brien}{O'Brien et~al.}{2016}]{OBrien2016}
O'Brien, T.~A., Kashinath, K., Cavanaugh, N.~R., Collins, W.~D., and O'Brien,
  J.~P. (2016).
\newblock A fast and objective multidimensional kernel density estimation
  method: {fastKDE}.
\newblock {\em Computational Statistics {\&} Data Analysis} {\bf 101,}
  148--160.

\bibitem[\protect\citeauthoryear{Robert}{Robert}{2010}]{robert_2010}
Robert, C.~P. (2010).
\newblock {\em Monte Carlo Statistical Methods (Springer Texts in Statistics)}.
\newblock Springer, paperback edition.

\bibitem[\protect\citeauthoryear{Shannon}{Shannon}{1948}]{Shannon1948}
Shannon, C.~E. (1948).
\newblock A mathematical theory of communication.
\newblock {\em Bell System Technical Journal} {\bf 27,} 379--423.

\bibitem[\protect\citeauthoryear{Silverman}{Silverman}{1986}]{Silverman1986}
Silverman, B.~W. (1986).
\newblock {\em Density Estimation for Statistics and Data Analysis}.
\newblock Chapman \& Hall, London.

\bibitem[\protect\citeauthoryear{Sklar}{Sklar}{1959}]{Sklar1959FonctionsDR}
Sklar, M.~J. (1959).
\newblock Fonctions de repartition a n dimensions et leurs marges.

\bibitem[\protect\citeauthoryear{Spirtes and Zhang}{Spirtes and
  Zhang}{2016}]{spirtes_2016}
Spirtes, P.~L. and Zhang, K. (2016).
\newblock Causal discovery and inference: concepts and recent methodological
  advances.
\newblock {\em Applied Informatics} {\bf 3,}.

\end{thebibliography}

\begin{figure*}
\centerline{\includegraphics[width = \textwidth]{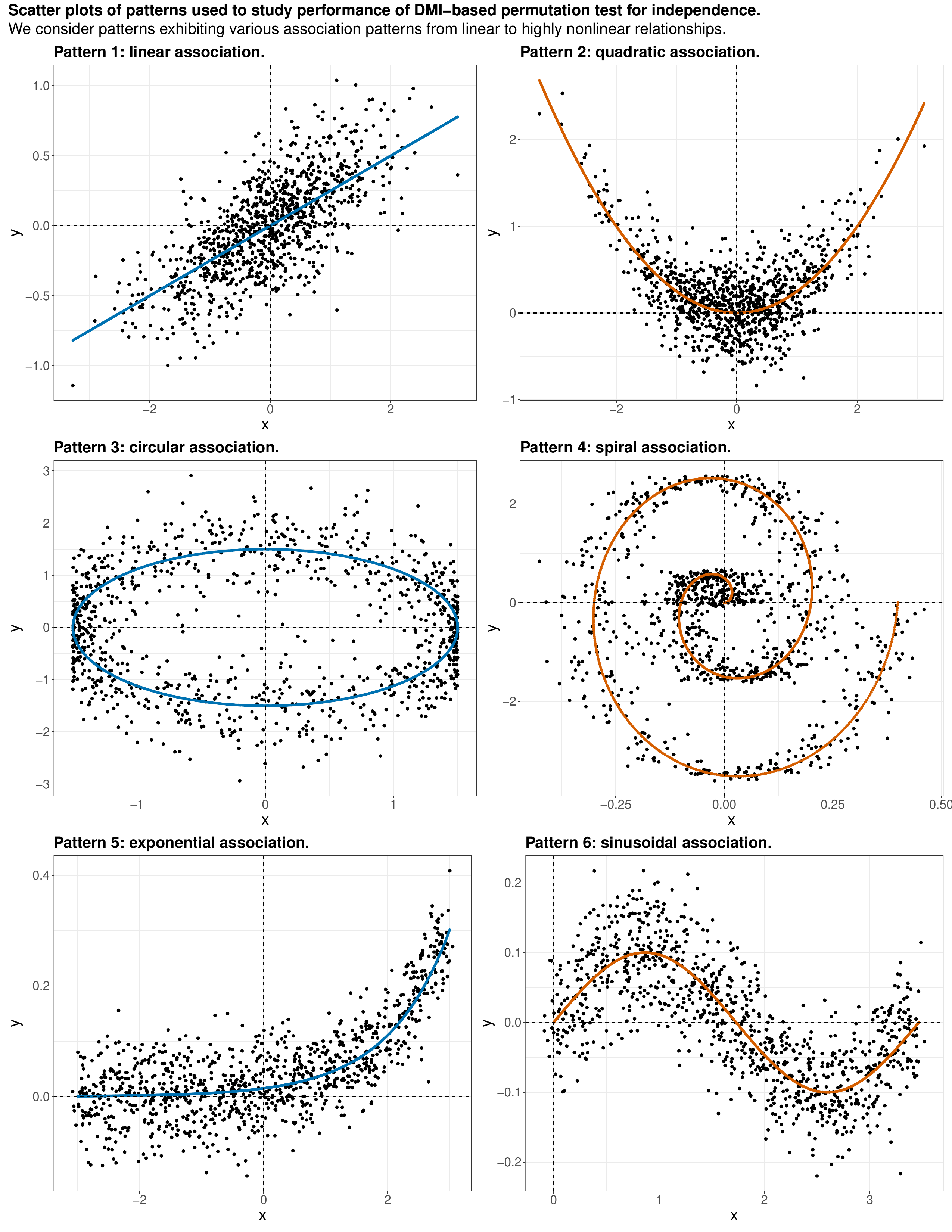}}
\caption{}
\label{fig:simulation1_scatter}
\end{figure*}

\end{document}